\documentclass[letterpaper, 10 pt, conference]{ieeeconf}  

\IEEEoverridecommandlockouts                              
\usepackage{graphics} 
\usepackage{subcaption}
\usepackage{amsmath} 
\usepackage{amssymb}  
\usepackage{tikz,pgf}  
\usepackage{algorithm,algorithmic}

\newtheorem{thm}{Theorem}[section]
\newtheorem{lem}[thm]{Lemma}
\newtheorem{prop}[thm]{Proposition}
\newtheorem{claim}[thm]{Claim}
\newtheorem{assum}{Assumption}[section]
\newtheorem{cor}[thm]{Corollary}

\newtheorem{exmp}{Example}[section]

    \def\ddefloop#1{\ifx\ddefloop#1\else\ddef{#1}\expandafter\ddefloop\fi}

    \def\ddef#1{\expandafter\def\csname c#1\endcsname{\ensuremath{\mathcal{#1}}}}
    \ddefloop ABCDEFGHIJKLMNOPQRSTUVWXYZ\ddefloop

    \def\ddef#1{\expandafter\def\csname s#1\endcsname{\ensuremath{\mathsf{#1}}}}
    \ddefloop ABCDEFGHIJKLMNOPQRSTUVWXYZ\ddefloop

    \def\ddef#1{\expandafter\def\csname b#1\endcsname{\ensuremath{\mathbb{#1}}}}
    \ddefloop ABCDEFGHIJKLMNOPQRSTUVWXYZ\ddefloop
    
\def\Ex{{\mathbf E}} 
\def\Pr{{\mathbf P}} 

\DeclareMathOperator*{\argmin}{arg\,min}

\title{\LARGE \bf
On The Role of Social Identity in the Market for (Mis)information
}

\author{Vijeth Hebbar and C\'edric Langbort
\thanks{This work was supported by the \textit{ARO} MURI grant W911NF-20-0252 (76582 NSMUR).}
\thanks{All authors are with the Coordinated Science Lab, University of Illinois at Urbana–Champaign,          Urbana, IL 61801, USA.
        {\tt\small \{vhebbar2,langbort\}@illinois.edu}}%
}

\begin{document}
















\maketitle
\thispagestyle{empty}
\pagestyle{empty}

\begin{abstract}
Motivated by recent works in the communication and psychology literature, we model and study the role social identity -- a person's sense of belonging to a group -- plays in human information consumption. A hallmark of Social Identity Theory (SIT) is the notion of `status', i.e., an individual's desire to enhance their and their `in-group's' utility \textit{relative} to that of an `out-group'. In the context of belief formation, this comes off as a desire to believe positive news about the in-group and negative news about the out-group, which has been empirically shown to support belief in misinformation and false news. 

We model this phenomenon as a Stackelberg game being played over an information channel between a news-source (sender) and news-consumer (receiver), with the receiver incorporating the `status' associated with social identity in their utility, in addition  to accuracy.  We characterize the strategy that must be employed by the sender to ensure that its message is trusted by receivers of all identities while maximizing their overall quality of information. We show that, as a rule, this optimal quality of information at equilibrium decreases when a receiver's sense of identity increases.  We further demonstrate how extensions of our model can be used to quantitatively estimate the level of importance given to identity in a population.
\end{abstract}

\section{INTRODUCTION}

Social Identity Theory (SIT), proposed by Tajfel and Turner \cite{Tajfel79}, seeks to formalize the notion that humans see the world through an `us' vs `them' lens. In doing so, it proposes the existence of three cognitive processes - categorization (breaking down society into categories e.g. races, genders), identification (assigning a set of categories for oneself and adopting pertinent norms) and comparison (comparing the categories assigned to oneself with others to determine one's status in the society). 

This last element of comparison hallmarks the theory of social identity and differentiates it from the idea of homophily (the tendency to form connections with similar individuals) in that, the status of an individual and their in-group is measured relative to an out-group. Under the lens of social identity theory, a decrease in the welfare of the out-group and an increase in welfare of the in-group both lead to an increased social status for the individual. Accordingly, any information that reflects negatively on the receiver's in-group affects their posterior beliefs differently than information that speaks positively of it. This variation in posteriors differs, however, from confirmation bias, as belief updates under the latter do not depend on the `nature' of the received information with respect to receiver's type, but only on its closeness to the their prior. 

Many models have been presented in a bid to capture the effect of identity in economic behaviour, in applications ranging from income redistribution and immigration \cite{Shayo09} to poverty and social exclusion \cite{Akerlof00}. 
Shayo \cite{Shayo09} proposes a game theoretic-model where the `identification' and `comparison' steps are taken simultaneously in equilibrium. In this framework, agents trade-off the relative `distance' they feel to an identity group with the social status they accrue from identifying with it. Akerlof and Kranton \cite{Akerlof00} consider a similar model but with the additional feature that payoffs from the actions depend explicitly on the identity group chosen.  In our work, we will assume that identities have already been assigned to agents and will seek to model the effect these identities have on the agent's beliefs. 

Identity is known to play a cornerstone role in belief of information, in particular, false news and misinformation. Bavel and Pereira \cite{Bavel18} provide empirical evidence suggesting that belief in misinformation is better explained by an identity driven processes than through confirmation bias or ideological value (belief in news that agree with one's ideologies) processes. Kahan \cite{Kahan16} theorises that belief in misinformation is driven by the psychological desire to signal and re-affirm ones identity, and thus status, in society. A prominent modern day example of socio-political identities driving belief in false news is seen among the deniers of climate change \cite{Bliuc15}. Thus, modeling the effect of social identity on misinformation can better inform our search for counter mechanisms to this ubiquitous issue. 

In our work, we propose a model with three elements -- source, sender and receiver -- acting over an information channel, with the source and the receiver having `types' (or in our context `identities'). Both the sender and receiver have their respective utilities and are thus incentivized to send and interpret information in strategic ways.  This approach is closely related to the idea of strategic information transmission \cite{Sobel82,Kamenica11}. 
More precisely, we view the interaction between sender and receiver as a Stackelberg game with the sender as the leader and the receiver as the follower. The role of social identity in this interaction is captured through the utilities of the agents.

Our work parallels the growing literature on (mis)information dynamics over networks, with many works using notions of antagonism and homophily among interacting agents to explain partisanship. The role of homophily in information propagation over networks, as well as extensions that incorporate misinformation explicitly have been studied recently \cite{Golub09,Acemoglu21}. Much existing work seeking to explain the creation and growth of polarization consider opinion dynamical models with antagonistic agents \cite{Hendrickx14}\textemdash\cite{Altafini13}. Leonard et al. \cite{Naomi21} consider opinion dynamical model incorporating political self-reinforcement (party polarization driven by internal mechanisms, e.g. eliminating moderating voices) and reflexive partisanship (party elites shunning a policy because the other side supports it) to explain the growth of polarization among political elite and validate their findings using historical data. Both the mechanisms considered in their work can be closely tied to the notion of political identity. Deviating from the standard DeGroot model for opinion dynamics that takes weighted averaging opinion updates,  Mei et al. \cite{Dorfler19} propose a weighted median opinion update model to better explain polarization and echo chamber formation.

Such studies over networks give us useful insights into the `virality' and persistence of inaccurate opinions and fake news, and can be used to arrive at counter-mechanisms to curb this spread. While all such models seek to understand how (mis)information (or the beliefs induced by them) spreads, we seek to look more closely at how an individual's beliefs are shaped by the presence of misinformation. In particular, our contribution lies in modeling how identity influences this process. Our work is closest in  spirit to the work by Mullainathan and Shleifer \cite{Mullainathan05} (which also inspired the reference to a 'market' in our paper's title) wherein, they try to model media slant stemming from readers' biases and study the effect of reader heterogeneity on the same. 

In Section \ref{sec:model}, we will outline the different elements of our model and introduce the utilities of the agents in our model. We also introduce an illustrative example to help us better interpret the results we obtain throughout this paper.  In Section \ref{sec:game_struct}, we will detail the structure of the game induced between the receiver and the sender and characterize the corresponding equilibrium. In Section \ref{sec:analytic_results} we will highlight and visualize some deeper implications of our model. We will see that these results act both to validate our model and to offer insight into a property counter-measures to curb belief in misinformation may possess. Finally, in Section \ref{sec:exploit}, we will augment our model with an observation mechanism allowing us to estimate parameters that define the receiver behaviour when they are a priori unknown. In doing so, we will illustrate the versatility of our model in being able to estimate an imperfectly characterized receiver's behaviour. 

This process of estimating a receiver's behavioural parameters from observations is tied to the broader literature of learning human preferences by observing their actions \cite{Tenenbaum09}. These approaches usually assume humans are Bayesian utility maximizing agents to learn their reward functions given a sequence of optimal actions. Such methods have found wide applicability in domains ranging from robotics \cite{Ng04} to natural language processing (NLP) \cite{Goodman13}. 

\section{Our Model} \label{sec:model}
We consider a communication theoretic model with three elements - source, sender and receiver. Associated with the source are two boolean random variables $(X,\theta)$, (we will call this the `state' and the `type' respectively) with some joint distribution $\pi$. For the purpose of this work, we will assume both $X\in\{0,1\}$ and $\theta\in\{A,B\}$ are Bernoulli random variables with a uniform distribution. We also assume that this distributional information is known to both the sender and the receiver. 

The sender encodes the information about the source (that is, $(X,\theta)$) into a Boolean `message' $Y\in\{a,b\}$ and sends it to the receiver. The receiver then decodes the message to make an `estimate' $\hat X$ on the true state of the source. The receiver also has a type $\bar\theta\in\{A,B\}$ which will influence how receiver chooses to decode the message. Figure \ref{fig:channel} depicts the standard information theoretic representation for our communication channel. 

The illustration also highlights the encoding and decoding strategies available to sender and receiver respectively. For instance, $m_A$ denotes the probability that the sender sends out the message $Y=a$ when the source has state $X=1$ and $p_{\bar{\theta}}$ signifies the probability that a receiver of type $\bar{\theta}$ infers $\hat X = 1$ from the message $Y=a$.

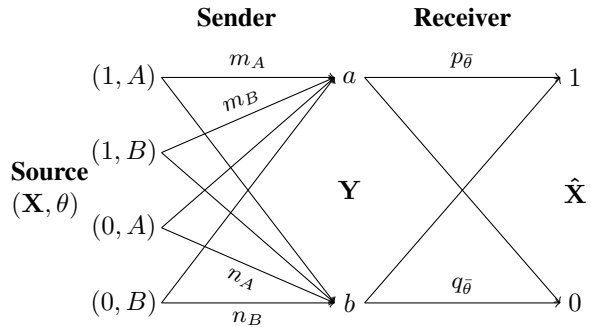
\begin{figure}[h]
    \centering



\begin{tikzpicture}

  \node at (-4,0) [text width=1cm]{\textbf{Source} $\displaystyle \mathbf{(X,\theta)}$};
  \node at (-3,1.5) {$\displaystyle (1,A)$};
  \node at (-3,0.5) {$\displaystyle (1,B)$};
  \node at (-3,-0.5) {$\displaystyle (0,A)$};
  \node at (-3,-1.5) {$\displaystyle (0,B)$};
  
  \node at (-1.5,2.3) [align=center, text width=2cm]{\textbf{Sender}};
  
  \draw [->] (-2.5,1.5) -- (-0.2,1.5) node [midway, sloped,above] {\small{$m_A$}};
  \draw [->] (-2.5,1.5) -- (-0.2,-1.5) node [midway, sloped,above] {};
  \draw [->] (-2.5,0.5) -- (-0.2,1.5) node [midway, sloped,above] {\small{$m_B$}};
  \draw [->] (-2.5,0.5) -- (-0.2,-1.5) node [midway, sloped,above] {};
  
  \draw [->] (-2.5,-1.5) -- (-0.2,1.5) node [midway, sloped,above] {};
  \draw [->] (-2.5,-1.5) -- (-0.2,-1.5) node [midway, sloped,below] {\small{$n_B$}};
  \draw [->] (-2.5,-0.5) -- (-0.2,1.5) node [midway, sloped,below] {};
  \draw [->] (-2.5,-0.5) -- (-0.2,-1.5) node [midway, sloped,below] {\small{$n_A$}};
  
  \node at (1.5,2.3) [align=center, text width=2cm]{\textbf{Receiver}};
  
  \node at (0,1.5) {$\displaystyle a$};
  \node at (0,0) {$\displaystyle \mathbf{Y}$};
  \node at (0,-1.5) {$\displaystyle b$};
  
  \draw [->] (0.2,1.5) -- (2.8,1.5) node [midway, sloped,above] {\small{$p_{\bar \theta}$}};
  \draw [->] (0.2,1.5) -- (2.8,-1.5) node [midway, sloped,above] {};
  
  \draw [->] (0.2,-1.5) -- (2.8,1.5) node [midway, sloped,above] {};
  \draw [->] (0.2,-1.5) -- (2.8,-1.5) node [midway, sloped,above] {\small{$q_{\bar \theta}$}};
  \node at (3,1.5) {$\displaystyle 1$};
  \node at (3,0) {$\displaystyle \mathbf{\hat{X}}$};
  \node at (3,-1.5) {$\displaystyle 0$};
  
  \end{tikzpicture}
  
    \caption{\textbf{(A classical communication channel representation of our model.)} The sender plays the role of the encoder while the receiver plays the role of a decoder.}
    \label{fig:channel}
\end{figure}

The types (both for the source and the receiver) will signify the identity of agents in our model. The way in which they impact the receiver's utility is best motivated in the context described below.

Let the source be a politician and let the state $X$ denote whether they are corrupt ($X=1$) or not $(X=0)$. The type $\theta$ then naturally indicates their political affiliation to one of two political groups $A$ or $B$. The sender can be viewed as a news agency who seeks to inform the receiver about the true nature of a politician through a news article (the `message'). Let us assume that the message $Y=a$ loosely indicates that the politician in question is corrupt while $Y=b$ signifies that the politician is not corrupt. Finally, the receiver may simply be an individual seeking to form an opinion on the politician. The receiver's type is an indication of their partisan nature and reflects their political identity. $\hat X=1$ ($\hat X=0$) indicates the receiver's opinion that the politician in question is corrupt (not corrupt). We will revisit this illustrative example repeatedly to highlight key aspects of our model as well as to elucidate the implications of our results.   




\subsection{Receiver Utility}
We model the receiver utility to include three components - accuracy utility, social identity utility and economic utility. Mathematically, for a receiver with type $\bar \theta$ we have
\begin{flalign*}
U^R(x,\hat x, \theta, \bar\theta) = \lambda_a^{\bar \theta} u^a(x,\hat x) + \lambda_s^{\bar \theta} u^s(\hat x, \theta, \bar\theta) + u^e(x),
\end{flalign*}
where $\lambda_a^\theta\in \bR_{\geq0}$ is the receivers accuracy weight and $\lambda_s\in \bR_{\geq0}$ is the receivers identity weight. The accuracy utility $u^a$, as the name suggests, offers a utility when the receiver's estimate of the state of the source is accurate. Mathematically,
\begin{flalign*}
    u^a(x,\hat x) = \delta_{x=\hat x}
\end{flalign*}
where $\delta$ is the usual indicator function. The social identity utility is mathematically given by
\begin{flalign*}
    u^s(\hat x, \theta, \hat \theta) = \begin{cases} 
    -\Delta_I^{\bar \theta} & \hat x=1, \theta = \bar \theta \\
    -\Delta_O^{\bar \theta} & \hat x=0, \theta \neq \bar \theta \\
    0 & \text{otherwise.}
    \end{cases}
\end{flalign*}

where $\Delta_O^{\bar\theta},\Delta_I^{\bar\theta}\in\bR_{\geq 0}$. The motivation for this utility is best seen by revisiting our illustrative example from Section \ref{sec:model}. $\Delta_I^{A}$, for instance, is the dis-utility to receiver of type $A$ when they believe that a politician of type $A$ is corrupt. Note that this utility does not depend on the true state of the source but only on the receiver's estimate of it and is thus, a psychological dis-utility. Likewise, $\Delta_O^{A}$ is the dis-utility to receiver of type $A$ when they believe that a politician of type $B$ is not corrupt. It is easy to see that this simple utility function captures the essence of social identity. Since the status of a group under Social Identity Theory is always seen as a relative metric, belief in both -- positive news about out-group politician and negative news about in-group politician -- can cause a decrease in in-group's social status.

Finally, we introduce the economic utility as
\begin{flalign*}
    u^e(x)=\delta_{x=0}
\end{flalign*}
which, in our politician story, simply indicates that every individual receives a utility from the politician in power being non-corrupt. Note that this utility depends only on the true state of the source and not on the receivers estimate. As such, this utility will be exogenous and independent of the receivers decoding strategy. Likewise, we could have included an additive component in the social identity utility that depends on the true state of the source, say an additional dis-utility if the true state of the in-group politician is corrupt. But once again, such a dis-utility will be exogenous to receiver's strategy and can be neglected for the purpose of our model. Going forward we will work with the truncated utility for the receiver as
\begin{flalign}
U^R(x,\hat x, \theta, \bar\theta) = \lambda_a^{\bar \theta} u^a(x,\hat x) + \lambda_s^{\bar \theta} u^s(\hat x, \theta, \bar\theta). \label{eq:rec_util}
\end{flalign}

\subsection{Sender Utility}
We consider the sender to require two things -- receiver's belief in their message and the receiver estimating the state of the true source accurately. We consider two receivers (one of either types) and we model the utility of the sender as
\begin{flalign}
    U^S(x,\hat x,\mathbf{p},\mathbf{q}) = \begin{cases}
        \delta_{x=\hat x} & \mathbf{p}>\mathbf{\frac{1}{2}}, \mathbf{q}>\mathbf{\frac{1}{2}}\\
        0 & \text{otherwise},
    \end{cases} \label{eq:sen_util}
\end{flalign}
where $\mathbf{p}=[p_A\; p_B]^T,\mathbf{q}=[q_A\; q_B]^T$ and $\mathbf{\frac{1}{2}}=[\frac{1}{2}\; \frac{1}{2}]^T$. We assume that if a receiver chooses not to believe in the sender ($p_{\bar\theta}\leq\frac{1}{2}$ or $ q_{\bar\theta}\leq\frac{1}{2}$), they will unsubscribe from the sender leading to loss in sender utility. In our illustrative example, an individual may choose to unsubscribe from a news agency if they stop believing in the veracity of the message $Y$. On the other hand, the sender also receives utility from the receiver being well-informed when they are `subscribed' to the sender and this is captured using the indicator function $\delta$ over the receiver estimate matching the true state of the source.  

\section{Game Over Information Channel} \label{sec:game_struct} 

Having laid out the receiver's and the sender's utilities and strategy sets, we now dwell into the information structure of our game. For now, we assume that the utility function of the receiver is common knowledge. The distribution on the true state and the type of the source is also common knowledge, but the type of the receiver is unknown to the sender. Before the game plays out, the sender commits to an encoding strategy and announces it to the receiver. This assumption of commitment is best understood in the context of our politician story, in which the sender (news agency) has to uphold a stance (a certain bias or lack thereof in reporting style). For instance, an openly partisan news agency would be expected to continue being so or risk losing credibility among its readership. 

\subsection{Receiver Optimal Strategy} \label{sec:R_opt_st}
For every message $Y$ generated by the sender strategy, the receiver picks their best response to maximize expected receiver utility 
\begin{flalign} 
    \max_{p_{\bar\theta},q_{\bar\theta}} \underbrace{\Ex[u^R(X,\hat X,\theta, \bar\theta)|Y]}_{\triangleq f(\bar\theta,Y)}, \label{eq:exp_rec_util}
\end{flalign}
where the expectation is taken over the distribution of messages $Y$ generated by the sender, distribution of the source state and type, and the distribution generated by the receiver strategy. Noting that the expectation above is linear in receiver strategy, it is easy to see that the optimal strategy for the receiver is given by\footnote{Strictly speaking, $p^*_{\bar\theta}$ is not unique when $\frac{\partial f(\bar\theta,a)}{\partial p}=0$ but we pick $p^*_{\bar\theta}=1$ as a convention. Likewise for $q^*_{\bar\theta}$.}
\begin{subequations}
\label{eq:opt_rec_strat}
\begin{flalign}
    p^*_{\bar\theta} & = \begin{cases}
        1 & \frac{\partial f(\bar\theta,a)}{\partial p}\geq0 \\
        0 & \frac{\partial f(\bar\theta,a)}{\partial p}<0,
    \end{cases} \label{eq:opt_rec_p}\\
    q^*_{\bar\theta} & = \begin{cases}
        1 & \frac{\partial f(\bar\theta,b)}{\partial q}\geq0 \\
        0 & \frac{\partial f(\bar\theta,b)}{\partial q}<0. \label{eq:opt_rec_q}
    \end{cases}
\end{flalign}
\end{subequations}
Thus, we see that the receiver employs pure strategies. 
\subsection{Sender Optimal Strategy} \label{sec:S_opt_st}
The sender on the other hand maximises its own expected utility assuming that the receiver is playing their best response. Mathematically, the sender picks its strategy to optimize
\begin{subequations}
\label{eq:complex_LP}
\begin{flalign}
    \max_{m_A,m_B,n_A,n_B\in[0,1]} \; & \Ex[\delta_{X=\hat X^*(Y)}] = \Pr(X=\hat X^*(Y)) \label{eq:complex_LP_obj}\\
    S.T \quad & \frac{\partial f(A,a)}{\partial p}\geq0, \quad \frac{\partial f(A,b)}{\partial q}\geq0 \label{eq:complex_LP_A_believes}\\
    & \frac{\partial f(B,a)}{\partial p}\geq0, \quad \frac{\partial f(B,b)}{\partial q}\geq0 \label{eq:complex_LP_B_believes}
\end{flalign}
\end{subequations}
where the expectation is taken over the distribution of the sources state and type, and over the distribution over messages generated by the sender's strategy. 


In (\ref{eq:complex_LP_obj}), $\hat{X}^*(Y)$ denotes the state estimate generated by the receiver's best response upon receiving message $Y$. Indeed, based on (\ref{eq:opt_rec_strat}), such a best response is nothing but a deterministic map from $\{a,b\}$ to $\{0,1\}$. (Simply, $a\mapsto \hat{X}^*(a)=1$ and $b\mapsto \hat{X}^*(b)=0$.) Also note that the constraints in (\ref{eq:complex_LP_A_believes}) and (\ref{eq:complex_LP_B_believes}) are a result of incorporating the receiver's best response into the sender utility, and act to ensure that the optimal receiver response is to always `believe' the sender. This results in a Stackelberg game between the sender (leader) and the receiver (follower), in accordance with the assumption of commitment on part of the sender.


With these elements in place, we can prove the following 
\begin{lem}
The optimization problem in (\ref{eq:complex_LP}) is an linear program in $m_A,m_B,n_A,n_B$.  Moreover, under our assumption of uniform distribution over the source states and types we can show that maximizing the expectation in (\ref{eq:complex_LP_obj}) is equivalent to maximizing the sum 
\begin{flalign}
    Q(m_A,n_A,m_B,n_B) \triangleq m_A+m_B+n_A+n_B, \label{eq:Q_defn}
\end{flalign}
where $Q$ denotes the `quality of information` available to the receiver. \label{lem:qual_info}
\end{lem}

Clearly if $Q=4$, then the sender is always truthful and the receiver can perfectly recreate the state of the source. On the other hand, if $Q<4$ the receiver loses information about the source due to the sender's encoding strategy.  
While this linear program can be solved using any one of the widely available numerical solvers, it can also be readily solved in closed form under an additional assumption on the receiver's identity parameters  ($\Delta_O^{\bar{\theta}}, \Delta_I^{\bar{\theta}}$) which we introduce in the next section. We see that our assumption offers us a large gain in analytical mileage while being only mildly restrictive. 

\subsection{Equilibrium in a Restricted Game}
Going forward we will make the following
\begin{assum} \label{ass:delta_rel}
    $$\Delta_O^{\bar\theta}\geq\Delta_I^{\bar\theta} \quad \forall \bar\theta \in \{A,B\}.$$
\end{assum}
Before making use of this assumption\footnote{We could have done a similar analysis if we instead had $\Delta_O^{\bar\theta}\leq\Delta_I^{\bar\theta} \quad \forall \bar\theta \in \{A,B\}$. The crucial element of our assumption is the existence of the same ordering between the two identity parameters holding true regardless of $\bar{\theta}$.} to simplify the linear program in (\ref{eq:complex_LP}), we note that it is qualitatively supported by a number of works in the Communication literature. Indeed, the study  \cite{Linden21} finds that negative language about out-group played a strong role in social media engagement through an empirical study on Facebook and Twitter. Moreover, works such as \cite{Etter11} have shown  that negative news content spreads faster and wider on social media.  While the intent to share a negative post about an out-group may not directly indicate the utilities gleaned from the process, we hypothesize that the relative gain in social status from doing so is higher than from sharing news about the in-group. 



Under Assumption \ref{ass:delta_rel} we can make the following 
\begin{prop} \label{prop:simplify_LP}
The optimal solution to (\ref{eq:complex_LP}) is given by $m_A^*=m_B^*=1$, and $(n_A^*,n_B^*)$ the optimal point of the linear program,
\begin{subequations} \label{eq:simple_LP}
\begin{flalign}
& \max_{n_A,n_B\in[0,1]}  n_A+n_B \\
     S.T \; & n_A(\lambda_s^A \Delta_I^A + \lambda_a^A)+n_B(\lambda_a^A-\lambda_s^A\Delta_O^A)\geq 0 \label{eq:A_believes}\\
      & n_B (\lambda_s^B \Delta_I^B + \lambda_a^B)+n_A(\lambda_a^B-\lambda_s^B\Delta_O^B)\geq 0 \label{eq:B_believes}
\end{flalign}
\end{subequations}
\end{prop}
First, note that the linear program in (\ref{eq:simple_LP}) always has a solution as $n_A=n_B=0$ is always a feasible point to the problem. The LP has two decision variables and two linear inequality constraints and can easily be solved analytically giving us
\begin{cor} \label{cor:simple_LP_soln}
The optimizer to (\ref{eq:simple_LP}) is given by,
\begin{flalign*}
    (n_A^*,n_B^*) = \begin{cases}
        (1,1) & k_A<0,k_B<0 \\
        (0,0) & k_B>k_A>0 \\
        (1,k_A) & 1>k_A>k_B \\
        (1,1) & k_A>1>k_B \\
        (\frac{1}{k_B},1) & k_A>k_B>1 \\
        \bigg(\min\{1,\frac{1}{k_B}\},1\bigg) & k_B>0>k_A 
    \end{cases}
\end{flalign*}
where
\begin{flalign}
    k_A \triangleq \frac{\lambda_s^A\Delta_I^A+\lambda_a^A}{\lambda_s^A\Delta_O^A-\lambda_a^A}, \;
    k_B \triangleq \frac{\lambda_s^B\Delta_O^B-\lambda_a^B}{\lambda_s^B\Delta_I^B+\lambda_a^B}. \label{eq:kA_kB_defn}
\end{flalign}
\end{cor}

The most important consequence of Corollary \ref{cor:simple_LP_soln} for our purposes is that the closed-form solution makes it possible to discuss the sender's optimal strategy qualitatively and, in particular, to study the influence of the various parameters accounting for the receiver's identity. 

\section{Analytical Results} \label{sec:analytic_results}
In this section, we will interpret some of the direct consequences of the equilibrium obtained in Corollary 3.3. In doing so, we will routinely revisit the illustrative example introduced in Section \ref{sec:model} to give more context and illustrate the implications of our results. First we have
\begin{cor} \label{cor:truth_negative_news}
    When $\Delta_O^{\bar\theta}>\Delta_I^{\bar\theta}$, the senders are \emph{always} completely truthful about `negative' states.
\end{cor}
Here, by `negative', we mean any state that causes loss in utility to the in-group, if realized. In our corrupt politician story, the `negative' state was `corrupt' ($X=1$) and so Corollary \ref{cor:truth_negative_news} implies that the news agency always truthfully reports when the source is corrupt. 

While Corollary \ref{cor:truth_negative_news} gives us a robust result on the nature of sender's equilibrium strategy, we would also like to know how the equilibrium varies with change in receiver characteristics (i.e change in $\lambda^{\bar\theta}$'s and $\Delta^{\bar\theta}$'s). This is the content of the following
\begin{claim} \label{cla:monotonicity}
    The quality of information $Q(m_A^*,m_B^*,n_A^*,n_B^*)$ available to receiver at equilibrium 
    \begin{itemize}
        \item decreases with increase in social identity levels ($\lambda_s^{\bar\theta}$) among the receivers,
        \item increases with increase in accuracy seeking levels ($\lambda_a^{\bar\theta}$) among the receivers, 
        \item increases with separation in relative importance of group identities ($\Delta_O^{\bar\theta}-\Delta_I^{\bar\theta}$) among the receivers. 
    \end{itemize}
\end{claim}

This claim very much agrees with intuition in the sense that a sender who seeks to maximise the quality of information available to receivers will do worse when the receivers themselves care more about their identity goals (and less about accuracy). While Claim \ref{cla:monotonicity} is a result of our modelling, it can also be viewed as a validation of our model as it agrees with empirical evidence on the role of social identity in misinformation \cite{Bavel18,Pennycook21}. Figure \ref{fig:opt_val_vs_lamb} illustrates the monotonic variation of quality of information with changes in social identity levels and accuracy seeking behaviour among receivers.  
\begin{figure}[ht]
    \centering
    \begin{subfigure}[b]{0.235\textwidth}
        \centering
        \includegraphics[width=\textwidth]{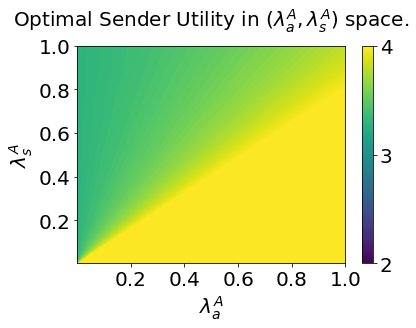}
        \caption{$\lambda_a^B=0.6,\lambda^B_s=0.4$}
        \label{fig:high_acc_lamb}
    \end{subfigure}
    \hfill
    \begin{subfigure}[b]{0.235\textwidth}
        \centering
        \includegraphics[width=\textwidth]{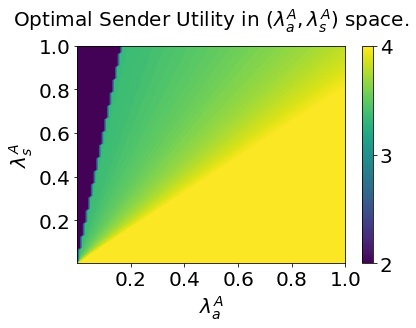}
        \caption{$\lambda_a^B=0.55,\lambda^B_s=0.45$}
        \label{fig:balance_acc_lamb}
    \end{subfigure}
    \hfill
    \begin{subfigure}[b]{0.235\textwidth}
        \centering
        \includegraphics[width=\textwidth]{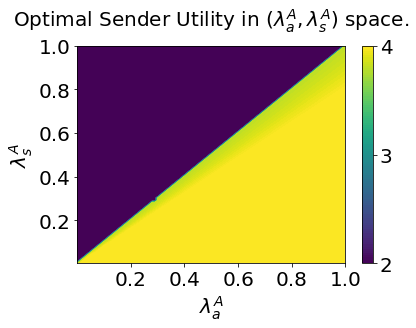}
        \caption{$\lambda_a^B=0.4,\lambda^B_s=0.6$}
        \label{fig:low_acc_lamb}
    \end{subfigure}
    \caption{\textbf{(Variation of quality of information with levels of identity.)} In the plots above, we have picked $\Delta_I^A=\Delta_I^B=1$, $\Delta_O^A=2$ and $\Delta_O^B=3.5$.}
    \label{fig:opt_val_vs_lamb}
\end{figure}

For fixed accuracy seeking and social identity levels in out-group (consider Figure \ref{fig:balance_acc_lamb}) we see that $Q$ drops sharply for high levels of in-group identity ($\lambda_s^A$) and low accuracy seeking levels($\lambda_a^A$). Further, with an increase in identity levels in an out-group receiver (going from Figure \ref{fig:balance_acc_lamb} to \label{ref:low_acc_lamb}) we see a sharp drop in $Q$ for some ranges of in-group identity levels. At the same time we see that in the general range where $\lambda_a^A>\lambda_s^A$ (yellow region in all plots in Figure \ref{fig:opt_val_vs_lamb}), increase in identity levels of an out-group receiver makes no difference to the quality of information available to the receivers. This shows that the sender can communicate a higher quality of information to a receiver of any type, even if only one of the identity groups cares enough about accuracy. In turn, interventions aimed at increasing the focus on accuracy (such as those proposed in \cite{Pennycook21,Thornhill19}) can be targeted at only one identity group.


\section{Exploiting Our Model} \label{sec:exploit}
We sought to model the effect of social identity with the specific aim of analyzing its role in misinformation belief.  In Section \ref{sec:analytic_results}, we obtained some results that aim to do so under the assumption that the sender has knowledge of the receiver payoffs. While we have made other assumptions, the knowledge of parameters defining the receiver utility ($\lambda_a^{\bar{\theta}},\lambda_s^{\bar{\theta}}, \Delta_I^{\bar{\theta}}$ and $\Delta_O^{\bar{\theta}}$) may be particularly difficult for the sender to obtain. 

In our politician story for instance, this simply means that news agencies do not have information about the identity and accuracy seeking levels among their audience. This would naturally be `learnt' over time through a sequence of `trial-and-error' attempts to know what their audience believes. In the same spirit, we will propose a thought experiment in Section \ref{sec:thought_exp} that, together with the structure of our problem, can be used to estimate these unknown quantities. These estimates can then be used by the sender to obtain its encoding strategy.  Going forward we will assume that the parameters defining receiver utility are unknown constants, but we will continue to work under Assumption \ref{ass:delta_rel}. 



\subsection{A Thought Experiment} \label{sec:thought_exp}



Not knowing the identity parameters of the receiver makes it impossible for the sender to solve (\ref{eq:simple_LP}) and obtain the optimal $n_A^*$ and $n_B^*$ for equilibrium. And so, in an attempt to estimate these unknowns, we postulate the existence of a mechanism for the sender to announce their strategy (i.e. $n_A$ and $n_B$, we already know $m_A=m_B=1$ for equilibrium)
to the receiver before the game is played out. The sender then responds whether or not they `believe' in the messages induced by this encoding strategy. A receiver of type $A$ (type $B$) believing in the message amounts to constraint (\ref{eq:A_believes}) (constraint (\ref{eq:B_believes}) respectively) being satisfied. Here we make the implicit assumptions that receiver's response is not noisy and can be used to accurately ascertain that the corresponding constraints are satisfied.

Now returning to the problem of estimation, we would like to design a procedure that takes as input a sequence of Boolean (`believe'/`don't believe') replies from the receiver and outputs the unknown parameters. Unfortunately, such a procedure would be ill-defined as the mapping from the set of identity parameters ($\lambda^{\bar\theta}$'s and $\Delta^{\bar\theta}$'s) to the sequence of replies is not one-to-one. However, the `augmented' identity parameters $k_A$ and $k_B$ as defined in (\ref{eq:kA_kB_defn}) can be identified uniquely. Fortunately, we can show that the knowledge of these augmented parameters is sufficient to ascertain receiver behaviour (i.e. to ascertain feasibility of (\ref{eq:A_believes}) and (\ref{eq:B_believes})) and, thus, for the sender to pick its optimal encoding strategy according to (\ref{eq:simple_LP}).

In Section \ref{sec:est_k}, we propose an algorithmic approach to estimating the augmented identity parameters ($k_A$ and $k_B$). In doing so, we also justify their sufficiency in evaluating receiver behaviour. Then in Section \ref{sec:use_k}, we use these estimates to arrive at sender's optimal encoding strategy. 

 
 \subsection{Estimating augmented identity parameters $k_A$ and $k_B$
 } \label{sec:est_k}
Note that we simply want feasible solutions to the set of linear inequalities given by (\ref{eq:A_believes}) and (\ref{eq:B_believes}). It is easy to see that if either $\lambda_s^A\Delta_O^A-\lambda_a^A\leq0$ or $\lambda_s^B\Delta_O^B-\lambda_a^B\leq0$ the corresponding receiver will believe the sender for any encoding strategy. In this case, we see that the end goal of obtaining the optimal sender strategy can be achieved even without explicitly computing the augmented identity parameters. 

So it suffices to consider the case when $\lambda_s^A\Delta_O^A-\lambda_a^A>0$ and $\lambda_s^B\Delta_O^B-\lambda_a^B>0$ and treat the former conditions as special cases. We can rework (\ref{eq:A_believes}) and (\ref{eq:B_believes}) to get, 
\begin{flalign}
    0 < k_B \leq \frac{n_B}{n_A}\leq k_A\label{eq:both_believe}
\end{flalign}
when $n_A\neq0$. Note that if $k_A<k_B$ then $n_A=n_B=0$ is the only feasible encoding strategy available to the sender. So we see that obtaining $k_A$ and $k_B$ is indeed sufficient to evaluate feasibility of an encoding strategy according to (\ref{eq:both_believe}). We now propose an algorithm that implements bisection search to obtain estimates for the augmented identity parameters. The algorithm takes two inputs - resolution $\delta>0$ to which we want our estimate to be accurate and upper bound $M>1$ on our search space for the parameters.
\begin{thm}
    The output $\hat{k}_A$ (and $\hat{k}_B$) of Algorithm $1$ satisfies $|k_A-\hat{k}_A|< \delta$ ($|k_B-\hat{k}_B|< \delta$) and is obtained in under $\big\lceil\log_2{\big(\frac{M}{\delta})}\big\rceil+1$ steps.   
\end{thm}

 \begin{algorithm}
 \caption{Algorithm for estimating identity parameter $k_A$ ($k_B$) when it is non-negative.}
 \begin{algorithmic}[1]
 \renewcommand{\algorithmicrequire}{\textbf{Input:}}
 \renewcommand{\algorithmicensure}{\textbf{Output:}}
 \REQUIRE $\delta>0,M>1$
 \ENSURE  $\hat{k}_A$ \\
    \textit{Initialization:}
    $\underline{k}_A^1=0, \bar{k}_A^1=M,\eta^1=1,i=1$ \\
     \WHILE{$|\bar{k}_A^i-\underline{k}_A^i|<\delta$}
        \STATE $n_B^i=\min\{1,\eta^i\}, n_A^i=\min\{1,\frac{1}{\eta^i}\}$ 
        \IF{(\ref{eq:A_believes}) \big((\ref{eq:B_believes}) for $k_B$\big) satisfied for this $n_A^i,n_B^i$} \label{lin:A_believes}
        \STATE $\underline{k}_A=\eta^i$     $\quad (\bar{k}_B=\eta^i)$ \label{lin:lower_thresh_update}
        \ELSE
        \STATE $\bar{k}_A=\eta^i$ $\quad (\underline{k}_B=\eta^i)$ \label{lin:upper_thresh_update}
        \ENDIF
        \STATE $\eta^{i+1}=\frac{\underline{k}_A+\bar{k}_A}{2}$ \label{lin:bisect}
        \STATE $i=i+1$
     \ENDWHILE
     \RETURN $\hat{k}_A=\frac{\underline{k}_A+\bar{k}_A}{2}$
 \end{algorithmic} 
 \label{alg:est_k}
 \end{algorithm}

To obtain $k_B$ we will replace every instance of $k_A$ by $k_B$ in Algorithm \ref{alg:est_k} and change lines \ref{lin:A_believes},\ref{lin:lower_thresh_update} and \ref{lin:upper_thresh_update} as indicated. 
\begin{exmp}
Let the true identity parameters for the receiver of either type be as given in Figure \ref{fig:balance_acc_lamb} (take $\lambda_s^A=0.45$, $\lambda_a^A=0.55$). Choosing  $M=10^{4}$ and $\delta=0.01$,  Algorithm \ref{alg:est_k} converges to estimates $\hat{k}_A=2.855$ and $\hat{k}_B=1.024$ in $21$ steps. True augmented parameters are $k_A=2.857$ and $k_B=1.025$. 
\end{exmp}

\subsection{Obtaining optimal encoding strategy for the sender} \label{sec:use_k}
For the purpose of this section we will treat the estimates $\hat{k}_A$ and $\hat{k}_B$ as the true augmented identity parameters. Note that the constraint in (\ref{eq:both_believe}) acts on the ratio of $n_A$ and $n_B$. Noting that both $n_A$ and $n_B$ are probabilities we can make the following 
\begin{claim} \label{cla:opt_encode_unknown_para}
    $n^*_A=n^*_B=0$ if $k_A<k_B$. Otherwise,
    \begin{subequations}
    \begin{flalign}
        & \max\{n^*_A,n^*_B\} =1 \label{eq:one_is_one}\\
        \text{and } & \frac{n^*_B}{n^*_A} =\argmin_{\gamma\in [\hat{k}_B,\hat{k}_A]} |\gamma-1|. \label{eq:gamma_near_1}
    \end{flalign}
    \label{eq:unknown_k_opt_str}
    \end{subequations}
\end{claim}
Recall that the sender is trying to maximize the quality of information available to the receivers and thus, wishes to pick $n_A$ and $n_B$ as close to one as possible. Conditions (\ref{eq:one_is_one}) and (\ref{eq:gamma_near_1}) together ensure that this indeed happens. Moreover, from (\ref{eq:gamma_near_1}) it is easy to see that the maximum value of Q (equal to 4) is attained when $1\in [k_B,k_A]$. 

Note that in Algorithm \ref{alg:est_k}, we do not explicitly check if $k_A>0$ and $k_B>0$. When $k_A$ ($k_B$) is negative our algorithm will return $\hat{k}_A=M$ ($\hat{k}_B=0$) and it can be shown that the optimal encoding strategy is still given by (\ref{eq:unknown_k_opt_str}).

\section{CONCLUSIONS AND FUTURE WORKS}

We set out to understand how social identity plays a role in formation and consumption of misinformation and proposed a communication model to achieve our goal. Despite its seeming simplicity, the predictions from our model are well aligned with empirical evidence and offer some useful insights in the search for counter-mechanisms to misinformation. We also showed how this model lends itself to estimation of some behavioural parameters of the receiver.


The work presented in Section \ref{sec:exploit} can be viewed as a parameter estimation problem. In doing so, we essentially considered a noiseless measurement setting in our thought experiment of Section \ref{sec:thought_exp}. A more realistic mechanism to infer whether receivers `believe' a certain sender strategy would likely involve `bandit feedback' instead,  whereby the various probabilities defining the sender's strategy, and the receiver's  'believe/not believe' choices, must be inferred from realized time series. Additionally, uncertainty may be introduced in this inference process `in two ways - the receiver may infer a different encoding strategy as generating the sampled messages and thus take actions responding to this strategy, or the receiver may simply respond sub-optimally to the `correct' strategy because they are boundedly rational. So further work needs to be done to address such `noisy' feedback.

In the same spirit, the natural next step in any control theoretic framework is to look at possible avenues of control. There is a vast literature that studies `nudging' - a form of `soft' control - to influence human choice architecture. We are currently looking to answer how `nudges' can be incorporated naturally into our framework. 

Finally, our model looks at a single receiver (having one of two types) with the idea that the receiver can be a prototype for a large population (of either type). This viewpoint can be enriched by considering a population of receivers with identity parameters picked from a distribution. We intend to look in this direction in subsequent work as well.

\appendices

\section{Proofs For Selected Results}
\subsection{Proof for Lemma \ref{lem:qual_info}}
From (\ref{eq:complex_LP}), we see that the sender wishes to maximize $P(X=\hat X^*(Y))$ given $p^*_{\bar\theta}=q^*_{\bar\theta}=1$. Let us look more closely at the objective function,
\begin{flalign*}
    P(X=\hat X) = &\sum_{\substack{x\in\{0,1\} \\\Theta\in\{A,B\} \\y\in\{a,b\}}} \Pr( X=x,\hat X=x, \theta=\Theta, Y=y) \\
    = & \sum_{\substack{x\in\{0,1\} \\\Theta\in\{A,B\} \\y\in\{a,b\}}} \Pr(\hat X = x|Y=y,X=x,\theta=\Theta) \\ & \Pr(Y=y|X=x,\theta=\Theta)  \Pr(X=x,\theta=\Theta). 
\end{flalign*}
Note that the receiver can only observe the message $Y$ in estimating $\hat X$ and so $\hat X$ in independent of $X$ and $\theta$ given $Y$. The mapping $Y$ to $\hat X$ is captured through the receivers decoding strategy i.e. through $p_{\bar\theta}$ and $q_{\bar\theta}$. The generation of message $Y$ from the source state $X$ and type $\theta$ is determined by the encoder strategy i.e. through $m_A,n_A,m_B$ and $n_B$. Finally, noting from our assumption that source state and types are uniformly distributed we get
\begin{flalign*}
    P(X=\hat X) = & \frac{1}{4} \big(p_{\bar\theta}(m_A+m_B)+(1-p_{\bar\theta})(2-m_A-m_B) \\ & +q_{\bar\theta}(n_A+n_B)+(1-q_{\bar\theta})(2-n_A-n_B)\big). 
\end{flalign*}
Incorporating receiver best response as $p^*_{\bar\theta}=1$ and $q^*_{\bar\theta}=1$ justifies the definition of $Q$ in Lemma \ref{lem:qual_info}. A similar expansion of (\ref{eq:exp_rec_util}) allows us to express the constraints in (\ref{eq:A_believes}) and (\ref{eq:B_believes}) as 
\begin{flalign}
    \lambda_s^A(\Delta_O^A&(m_B+1-n_B)-\Delta_I^A(m_A+1-n_A)) \notag\\
    &+\lambda_a^A(n_A+n_B+m_A+m_B-2)\geq 0 \label{eq:A_believes_a}\\
    \lambda_s^A(\Delta_I^A&(n_A+1-m_A)-\Delta_O^A(n_B+1-m_B)) \notag\\
    &+\lambda_a^A(n_A+n_B+m_A+m_B-2)\geq 0 \label{eq:A_believes_b}\\
    \lambda_s^B(\Delta_O^B&(m_A+1-n_A)-\Delta_I^B(m_B+1-n_B)) \notag\\
    &+\lambda_a^B(n_A+n_B+m_A+m_B-2)\geq 0 \label{eq:B_believes_a}\\
    \lambda_s^B(\Delta_I^B&(n_B+1-m_B)-\Delta_O^B(n_A+1-m_A)) \notag\\
    &+\lambda_a^B(n_A+n_B+m_A+m_B-2)\geq 0. \label{eq:B_believes_b}
\end{flalign}
So we see that optimization (\ref{eq:complex_LP}) is indeed a linear program in $m_A,n_A,m_B$ and $n_B$. 

\addtolength{\textheight}{-0.3cm}   


\subsection{Proof for Proposition \ref{prop:simplify_LP}}
We consider the linear program form for (\ref{eq:complex_LP}) with (\ref{eq:Q_defn}) as the objective function and constraints (\ref{eq:A_believes_a})\textemdash(\ref{eq:B_believes_b}), and show that it can be simplified to (\ref{eq:simple_LP}). First, we consider (\ref{eq:A_believes_a}) and see that,
\begin{flalign*}
    &\lambda_s^A(\Delta_O^A(m_B+1-n_B)-\Delta_I^A(m_A+1-n_A)) \\
    = &\lambda_s^A(\Delta_I^A(n_A+1-m_A)-\Delta_O^A(n_B+1-m_B)) \\ &+\underbrace{2\lambda_s^A(\Delta_O^A-\Delta_I^A)}_{\geq 0},
\end{flalign*}
which shows that (\ref{eq:A_believes_b}) implies (\ref{eq:A_believes_a}) under Assumption \ref{ass:delta_rel}. Similarly we can show that (\ref{eq:B_believes_b}) implies (\ref{eq:B_believes_a}). So it suffices to solve the LP with objective function (\ref{eq:Q_defn}) and constraints (\ref{eq:A_believes_b}) and (\ref{eq:B_believes_b}). For convenience we redefine the left-hand side in (\ref{eq:A_believes_b}) and (\ref{eq:B_believes_b}) as $g_A(m_A,n_A,m_B,n_B)$ and $g_A(m_A,n_A,m_B,n_B)$ respectively. 

Also, note that $m_A=m_B=1$ and $n_A=n_B=0$ is a feasible solution to (\ref{eq:A_believes_b}) and (\ref{eq:B_believes_b}) and so, the value of our LP is lower bounded by $2$ giving us,
\begin{cor} \label{cor:LP_lower_bound}
    $m_A^*+n_A^*+m_B^*+n_B^* \geq 2$
\end{cor}
where $(m_A^*,n_A^*,m_B^*,n_B^*)$ denotes the optimal encoding strategy and will be used as such throughout our proof. In continuing our proof, we first state

\begin{lem}
\label{lem:n_pos_implies_m_1}
\begin{flalign}
n_B^*>0 \implies m_A^*=1 \label{eq:nB_pos_implies_mA_1}\\
n_A^*>0 \implies m_B^*=1 \label{eq:nA_pos_implies_mB_1}
\end{flalign}
\end{lem}
\begin{proof}
We prove  here, (\ref{eq:nB_pos_implies_mA_1}) and (\ref{eq:nA_pos_implies_mB_1}) can be shown similarly. Let us assume instead that $m_A^*<1$ and $n_B^*>0$ and define
\begin{flalign*}
    \zeta \triangleq \min\bigg\{n_B^*, \frac{(\lambda_s^A \Delta_I^A-\lambda_a^A)}{(\lambda_s^A \Delta_O^A-\lambda_a^A)}(1-m_A^*)\bigg\} > 0 
\end{flalign*}
Now consider an alternate set of decisions variables obtained as,
\begin{flalign*}
    &(m_A',n_A',m_B',n_B') \\ = &\bigg(m_A^*+\frac{(\lambda_s^A \Delta^A_O-\lambda_a^A)}{(\lambda_s^A \Delta^A_I-\lambda_a^A)}\zeta,n_A^*,m_B^*,n_B^*-\zeta\bigg)
\end{flalign*}
The definition of $\zeta$ ensures that $m_A',n_B'\in[0,1]$. We will now evaluate the feasibility of this new set of points.
\begin{flalign*}
    & g_A(m_A',n_A',m_B',n_B') = g_A(m_A^*,n_A^*,m_B^*,n_B^*) \\  &+(\lambda_s^A \Delta^A_I-\lambda_a^A) \frac{(\lambda_s^A \Delta^A_O-\lambda_a^A)}{(\lambda_s^A \Delta^A_I-\lambda_a^A)}\zeta-(\lambda_s^A \Delta^A_O-\lambda_a^A)\zeta \\
    &= g_A(m_A^*,n_A^*,m_B^*,n_B^*) \geq 0, \\
    & g_B(m_A',n_A',m_B',n_B') = g_B(m_A^*,n_A^*,m_B^*,n_B^*) \\ & +(\lambda_s^B \Delta_O^B+\lambda_a^B) \frac{(\lambda_s^A \Delta_O^A-\lambda_a^A)}{(\lambda_s^A \Delta^A_I-\lambda_a^A)}\zeta-(\lambda_s^B \Delta_I^B+\lambda_a^B)\zeta \\
    & = g_B(m_A^*,n_A^*,m_B^*,n_B^*) + (\lambda_s^B \Delta_I^B+\lambda_a^B)  \\ 
    & \times \bigg(\underbrace{\frac{(\lambda_s^B \Delta_O^B+\lambda_a^B)}{(\lambda_s^B \Delta_I^B+\lambda_a^B)}}_{> 1\; \because \Delta^B_O>\Delta^B_I} \times \underbrace{\frac{(\lambda_s^A \Delta_O^A-\lambda_a^A)}{(\lambda_s^A \Delta_I^A-\lambda_a^A)}}_{> 1\; \because \Delta^A_O>\Delta^A_I}-1\bigg)\zeta \\
    &>g_B(m_A^*,n_A^*,m_B^*,n_B^*) \geq 0.
\end{flalign*}
Note the explicit usage of Assumption \ref{ass:delta_rel} above. Finally we compute the value of objective as, 
\begin{flalign*}
 & Q(m_A',n_A',m_B',n_B') = Q(m_A^*,n_A^*,m_B^*,n_B^*) + \\ & \underbrace{\frac{(\lambda_s^A \Delta_O-\lambda_a^A)}{(\lambda_s^A \Delta_I-\lambda_a^A)}}_{>1 \; \because \Delta^A_O>\Delta^A_I}\zeta - \zeta > Q(m_A^*,n_A^*,m_B^*,n_B^*). 
\end{flalign*}
The last inequality above again relies on Assumption \ref{ass:delta_rel} and contradicts the fact that we picked $(m^*_A,m_B^*,n_A^*,n_B^*)$ to be an optimizer. And so, $\zeta$ must be zero. A direct implication from the definition of $\zeta$ is then that if $n_A^*>0$ we must have $m_B^*=1$. 
\end{proof}

We are now in position to prove Proposition \ref{prop:simplify_LP}. First, we show that both $m_A^*$ and $m_B^*$ are equal to $1$. We will split this in two different cases.
\begin{itemize}
    \item \textbf{Case I:($(\lambda_a^A - \lambda_s^A \Delta_O^A)<0$ and $(\lambda_a^B - \lambda_s^B \Delta_O^B)<0$)} Let us assume WLOG that $m_B^*<1$. By Lemma \ref{lem:n_pos_implies_m_1} we must have $n_A^*=0$ as well. From Corollary \ref{cor:LP_lower_bound} we also know that $m_A^*+n_B^*>1$. For this to hold, once again invoking Lemma \ref{lem:n_pos_implies_m_1} we must have $m_A*=1,n_B>0$. For our point to be feasible it must satisfy the constraint,
    \begin{flalign}
         0  \leq  g_A(1,0,m_B,n_B) \notag 
        = & -(\lambda_a^A + \lambda_s^A \Delta_O^A) (1-m_B) \\
        &
        + (\lambda_a^A - \lambda_s^A \Delta_O^A) n_B. \label{eq:mA_mB_1_proof_I}
    \end{flalign}
    It is then easy to see that we must have $m_B^*=1,n_B^*=0$ when $\lambda_a^A - \lambda_s^A \Delta_O^A)<0$ to satisfy Equation (\ref{eq:mA_mB_1_proof_I}). This is in contradiction with our assumption of $m_B^*<1$, so $m_B^*$ must be $1$. 
    \item \textbf{Case II:($(\lambda_a^A - \lambda_s^A \Delta_O^A) \geq 0$ or $(\lambda_a^B - \lambda_s^B \Delta_O^B) \geq 0$)}. WLOG assume $(\lambda_a^A - \lambda_s^A \Delta_O^A)\geq 0$ then since $\Delta_O^A>\Delta_I^A$ we also have $(\lambda_a^A - \lambda_s^A \Delta_I^A)>0$. Then we observe that $g_A,g_B$ and $f$ are increasing in $m_A$ and $n_B$. Then it is easy to see that $m_A^*=1$ and $n_B*=1$. Now suppose $m_B^*<1$. Recall that this then means $n_A^*=0$ by Lemma \ref{lem:n_pos_implies_m_1}. Once again we obtain the constraints to be satisfied as,
    \begin{flalign*}
        0\leq g_A(1,0,m_B^*,1) = (\lambda_a^A + \lambda_s^A \Delta_O^A) m_B^* - 2\lambda_s^A \Delta_O^A \notag \\
        0\leq g_B(1,0,m_B^*,1) = (\lambda_a^B - \lambda_s^B \Delta_I^B) m_B^* + 2 \lambda_s^B \Delta_I^B  \notag 
    \end{flalign*}
    We see that $g_B\geq0$ holds for all values of $m_B^*$ but to ensure that $g_A\geq 0$ for all values of $m_B^*$ we must have 
    \begin{flalign*}
        m_B^* \geq \frac{2\lambda_s^A \Delta_O^A}{(\lambda_a^A + \lambda_s^A \Delta_O^A) }
    \end{flalign*}
    for all values of $\Delta_O^A,\lambda_a^A$ and $\lambda_s^A$ such that $\lambda_a^A - \lambda_s^A \Delta_O^A \geq 0$. Taking $\lambda_a^A - \lambda_s^A \Delta_O^A = \epsilon \geq 0$ we then have  
    \begin{flalign*}
        m_B^*  \geq \frac{2(\lambda_a^A -\epsilon)}{(2\lambda_a^A - \epsilon)} \quad \forall \epsilon \geq 0, \quad \forall  \lambda_a^A
        \implies m_B^*  \geq 1 
    \end{flalign*}
    We arrive at a contradiction. So we must have $m_B^* = 1$ and Proposition \ref{prop:simplify_LP} is proven. 
\end{itemize}

\section{Illustrating Algorithm \ref{alg:est_k}}
Let us suppose WLOG that the true $k_A>1$. We sequentially improve the lower bound $\underline{k}_A$ and the upper bound $\bar{k}_A$ for our estimate $\hat{k}_A$ using our `observations'. We initialize $\underline{k}_A^1=0$ and $\bar{k}_A^1=\infty$. We will call $[\underline{k}^i_A,\bar{k}^i_A]$ our region of interest and sequentially shrink this region until it has Lebesgue measure less than $\delta$.

We will aim to find $n_A$ and $n_B$ that satisfies (\ref{eq:both_believe}) and we note that only the ratio $\frac{n_B}{n_A}$ needs to satisfy certain properties to do so. Algorithm \ref{alg:est_k} in fact iterates over values of this ratio (which we indicate by $\eta$). Initializing $\eta^1=1$, we ask receiver A if they will believe the sender for the corresponding encoding strategy.  We see in Figure \ref{fig:illus_alg_1} that receiver $A$ indeed believes the sender and so our estimate $\hat{k}_A>\eta^1$. Accordingly we update the lower threshold as $\underline{k}_A^2=\eta^1$.

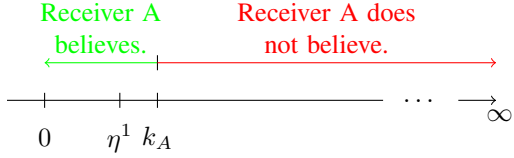
\begin{figure}[h]
    \centering
    \begin{tikzpicture}

    \draw [-] (-0.5,0) -- (4.5,0) node [midway, sloped,above] {};
    \node at (5,0) {$\displaystyle \dots$};
    \draw [->] (5.5,0) -- (6,0) node [midway, sloped, below right] {$\infty$};
     
    \draw [-] (0,0.1) -- (0,-0.1); 
    \node at (0,-0.5) {$\displaystyle 0$};
     
    \draw [-] (1,0.1) -- (1,-0.1); 
    \node at (1,-0.5) {$\displaystyle \eta^1$};
     
    \draw [-] (1.5,0.1) -- (1.5,-0.1); 
    \node at (1.5,-0.5) {$\displaystyle k_A$};
     
    \draw [->,green] (1.5,0.5) -- (0,0.5) node [text width=0.1\textwidth, text centered, midway, sloped,above] {\textcolor{green}{Receiver A believes.}};
    \draw [-] (1.5,0.4) -- (1.5,0.6); 
    \draw [->,red] (1.5,0.5) -- (6,0.5) node [text width=0.2\textwidth,text centered, midway, sloped,above] {\textcolor{red}{Receiver A does not believe.}};

  \end{tikzpicture}
  
    \caption{\textbf{(First iteration of Algorithm \ref{alg:est_k}.)} The first guess for encoding strategy lies in the region where the receiver A believes.}
    \label{fig:illus_alg_1}
\end{figure}

We now generate our next guess for $\eta_2$ by bisecting the new region of interest according to Line \ref{lin:bisect}. Note that the upper threshold is currently $\bar{k}_A^2=\infty$. In reality, we will use some large and finite value for this upper threshold so this bisection step is well-defined. Having obtained the next guess $\eta^2$, we will again evaluate whether the receiver believes the sender. From Figure \ref{fig:illus_alg_2} we see that the receiver does not and so we re-set the upper threshold as $\bar{k}_A^3=\eta^2$.    
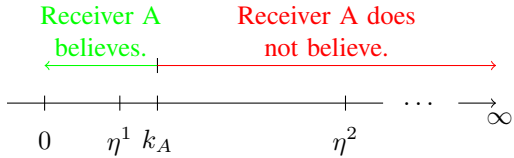
\begin{figure}[h]
    \centering
    \begin{tikzpicture}

    \draw [-] (-0.5,0) -- (4.5,0) node [midway, sloped,above] {};
    \node at (5,0) {$\displaystyle \dots$};
    \draw [->] (5.5,0) -- (6,0) node [midway, sloped, below right] {$\infty$};
     
    \draw [-] (0,0.1) -- (0,-0.1); 
    \node at (0,-0.5) {$\displaystyle 0$};
     
    \draw [-] (1,0.1) -- (1,-0.1); 
    \node at (1,-0.5) {$\displaystyle \eta^1$};
    
    \draw [-] (4,0.1) -- (4,-0.1); 
    \node at (4,-0.5) {$\displaystyle \eta^2$};
     
    \draw [-] (1.5,0.1) -- (1.5,-0.1); 
    \node at (1.5,-0.5) {$\displaystyle k_A$};
     
    \draw [->,green] (1.5,0.5) -- (0,0.5) node [text width=0.1\textwidth, text centered, midway, sloped,above] {\textcolor{green}{Receiver A believes.}};
     \draw [-] (1.5,0.4) -- (1.5,0.6); 
    \draw [->,red] (1.5,0.5) -- (6,0.5) node [text width=0.2\textwidth,text centered, midway, sloped,above] {\textcolor{red}{Receiver A does not believe.}};

  \end{tikzpicture}
  
    \caption{\textbf{(Second iteration of Algorithm \ref{alg:est_k}.)} The second guess for encoding strategy lies in the region where the receiver A does not believes.}
    \label{fig:illus_alg_2}
\end{figure}

We carry on bisecting our regions of interest to obtain subsequent guesses $\eta^i$ and the algorithm terminates when the region of interest $[\underline{k}_A^i,\bar{k}_A^i]$ is small enough.

\begin{thebibliography}{99}

\bibitem{Tajfel79}
H. Tajfel and J. C. Turner, ``An integrative theory of inter-group conflict.'' {\it In W. G. Austin \& S. Worchel (Eds.), The social psychology of inter-group relations}, Monterey, CA: Brooks/Cole, 1979, pp 33-47.

\bibitem{Shayo09}
M. Shayo, ``Model of Social Identity with an Application to Political Economy: Nation, Class, and Redistribution." {\it The American Political Science Review}, Monterey, CA: Brooks/Cole, vol. 103, no. 2, 2009, pp 147-174.

\bibitem{Akerlof00}
G. A. Akerlof, R. E. Kranton, ``Economics and Identity'', \textit{The Quarterly Journal of Economics}, vol. 115, iss. 3, 2000, pp 715–753, doi: 10.1162/003355300554881.

\bibitem{Bavel18}
J.J. Van Bavel, A. Pereira, ``The Partisan Brain: An Identity-Based Model of Political Belief.'' \textit{Trends Cogn Sci.}, vol. 22, no. 3, 2018, pp 213-224. doi: 10.1016/j.tics.2018.01.004. 

\bibitem{Kahan16}
D. M. Kahan, ``The Politically Motivated Reasoning Paradigm, Part 1: What Politically Motivated Reasoning Is and How to Measure It.", {\it Emerging trends in the Social and Behavioural Sciences}, 2016, doi: 10.1002/9781118900772.etrds0417.

\bibitem{Bliuc15}
 A. M. Bliuc, C. McGarty, E. Thomas et al. ``Public division about climate change rooted in conflicting socio-political identities.'' \textit{Nature Clim Change}, vol. 5, pp. 226–229, 2015, doi: 10.1038/nclimate2507.

\bibitem{Sobel82}
 V. P. Crawford and J. Sobel, ``Strategic Information Transmission.” \textit{Econometrica}, vol. 50, no. 6, 1982, pp. 1431–51, doi:10.2307/1913390.

\bibitem{Kamenica11}
E. Kamenica and M. Gentzkow, ``Bayesian Persuasion." \textit{American Economic Review}, vol. 101, no. 6, 2011, pp. 2590-2615, doi: 10.1257/aer.101.6.2590.


\bibitem{Golub09}
B. Golub and M. O. Jackson, ``How Homophily Affects Diffusion and Learning in Networks.'' \textit{arXiv preprint}, arXiv:0811.4013, 2009.  

\bibitem{Acemoglu21}
D. Acemoglu, A. Ozdaglar and J. Siderius, "A Model of Online Misinformation", \textit{NBER Working Paper}, no. 28884, 2021, doi = 10.3386/w28884.

\bibitem{Hendrickx14}
J. M. Hendrickx, "A lifting approach to models of opinion dynamics with antagonisms," 53rd IEEE Conference on Decision and Control, 2014, pp. 2118-2123, doi: 10.1109/CDC.2014.7039711.

\bibitem{Cao16}
A. V. Proskurnikov, A. S. Matveev and M. Cao, "Opinion Dynamics in Social Networks With Hostile Camps: Consensus vs. Polarization," in IEEE Trans. Autom. Control, vol. 61, no. 6, 2016, pp. 1524-1536, doi: 10.1109/TAC.2015.2471655.


\bibitem{Altafini13}
C. Altafini, "Consensus problems on networks with antagonistic interactions", IEEE Trans. Autom. Control, vol. 58, no. 4, pp. 935-946, 2013

\bibitem{Naomi21}
N. E. Leonard, K. Lipsitz, A. Bizyaeva, A. Franci  and Y. Lelkes, ``The nonlinear feedback dynamics of asymmetric political polarization." \textit{Proceedings of the National Academy of Sciences}, vol. 118, no. 50, 2021, doi: 10.1073/pnas.2102149118.

\bibitem{Dorfler19}
W. Mei, F. Bullo, G. Chen and F. Dörfler,  ``Occam’s razor in opinion dynamics: The weighted-median influence process.'' \textit{arXiv preprint}, arXiv:0811.4013, 2019  

\bibitem{Mullainathan05}
S. Mullainathan and A. Shleifer. “The Market for News.” \textit{American Economic Review}, vol. 95, no. 4, 2005, pp. 1031-1053, doi: 10.1257/0002828054825619. 

\bibitem{Tenenbaum09}
C. L. Baker, R. Saxe and J. B. Tenenbaum, "Action understanding as inverse planning", Cognition, vol. 113, no. 3, pp. 329-349, 2009.

\bibitem{Ng04}
P. Abbeel and A. Y. Ng. ``Apprenticeship learning via inverse reinforcement learning.'' In Proc. of the 21st International Conference on Machine Learning, 2004, doi:10.1145/1015330.1015430.

\bibitem{Goodman13}
N.D. Goodman, and A. Stuhlmüller. "Knowledge and implicature: Modeling language understanding as social cognition." \textit{Topics in cognitive science}, vol. 5, no. 1, 2013, pp. 173--184.

\bibitem{Linden21}
S. Rathje, J.J. Van Bavel, S. van der Linden, ``Out-group animosity drives engagement on social media.'' Proc Natl Acad Sci USA, vol. 118, no. 26.,2021, doi: 10.1073/pnas.2024292118.

\bibitem{Etter11}
L.K. Hansen, A. Arvidsson,F.A. Nielsen, E. Colleoni E., M. Etter, ``Good Friends, Bad News - Affect and Virality in Twitter,'' in Park J.J., Yang L.T., Lee C. (eds) Future Information Technology. \textit{Communications in Computer and Information Science}, vol. 185. Springer, Berlin, Heidelberg, 2011, doi:10.1007/978-3-642-22309-9\_5

\bibitem{Pennycook21}
G. Pennycook, Z. Epstein, M. Mosleh, A. A. Arechar, D. Eckles and D. G. Rand, ``Shifting attention to accuracy can reduce misinformation online." {\it Nature}, vol. 592, 2021, pp 590-595, doi: 10.1038/s41586-021-03344-2.

\bibitem{Thornhill19}
C. Thornhill, Q. Meeus, J. Peperkamp, B. Berendt, ``A Digital Nudge to Counter Confirmation Bias.'' \textit{Front Big Data}, vol.2, no. 11, 2019, doi: 10.3389/fdata.2019.00011



 

\end{thebibliography}
\end{document}